\documentclass[a4paper,UKenglish,cleveref, autoref]{lipics-v2019}

\usepackage{xargs}
\usepackage{amsthm}
\usepackage{amsmath}
\usepackage{nccmath}
\usepackage{oxford}
\newcommand{\ignore}[1]{}

\newcommand{\SetN}{\mathbb{N}}
\newcommand{\SetZ}{\mathbb{Z}}
\newcommand{\SetR}{\mathbb{R}}

\newcommand{\card}[1]{\mthempty{\argint{\vert}{#1}{\vert}}}

\newcommand{\LTL}{\mthfun{LTL}\xspace}
\newcommand{\GRone}{\mthfun{GR(1)}\xspace}
\def\TEMPORAL#1{\mbox{\small\boldmath$\mathbf{#1}$}}
\def\ltlnext{\TEMPORAL{X}}
\def\sometime{\TEMPORAL{F}} 
\def\always{\TEMPORAL{G}}
\def\until{\,\TEMPORAL{U}\,}
\def\Nat{\mathbb{N}}

\newcommand{\Ag}{\mthset{N}}
\newcommand{\Ac}{\mthset{Ac}}
\newcommand{\AcProf}{\vec{\Ac}}
\newcommand{\St}{\mthset{St}}
\newcommand{\AP}{\mthset{AP}}

\renewcommand{\Game}{\mthname{G}}

\newcommand{\pun}{\mthsym{pun}}

\newcommand{\labFun}{\lambda}

\newcommand{\trnFun}{\mthfun{tr}}
\newcommand{\act}{\mthsym{a}}
\newcommand{\jact}{\vec{\mthsym{a}}}

\newcommand{\StrSet}{\mthset{Str}}

\newcommand{\strElm}{\sigma}

\newcommand{\strpElm}{\mthelm{\vec{\sigma}}}

\newcommand{\NE}{\mthset{NE}}

\newcommand{\winsym}{\mthset{Win}}
\newcommandx{\Win}[3][1=, 2=, 3=]
{\mthset{\winsym#3}[#1][#2]}

\newcommand{\presym}{\mthfun{Pre}}
\newcommandx{\Pre}[3][1=, 2=, 3=]
{\mthset{\presym#3}[#1][#2]}

\newcommand{\eqsym}{\mthfun{Eq}}
\newcommandx{\Eq}[3][1=, 2=, 3=]
{\mthset{\eqsym#3}[#1][#2]}

\newcommand{\WI}{\mthset{WI}}

\newcommand{\SI}{\mthset{SI}}

\newcommand{\K}{\mathcal{K}}

\newcommandx{\AFW}[5][1=, 2=, 3=, 4=, 5=]
{\txtargname{AFW#5{\small\argint{$[$}{#1}{$]$}}}[#2][#3]{#4}\xspace}

\newcommand{\MP}[1][]{%
	\ifthenelse{\equal{#1}{}}{{\small{\sf mp}}}{{\small{\sf mp}}(#1)}%
\xspace}

\providecommand{\strFun}[1][]{\mthfun{\sigma}}
\providecommand{\pstrFun}[1][]{\mthfun{\sSym}}

\newcommand{\weak}{{\normalfont\textsc{Weak Implementation}}\xspace}
\newcommand{\weakc}{{\normalfont\textsc{Weak Implementation Complement}}\xspace}
\newcommand{\strong}{{\normalfont \textsc{Strong Implementation}}\xspace}
\newcommand{\strongc}{{\normalfont \textsc{Strong Implementation Complement}}\xspace}
\newcommand{\optwi}{{\normalfont \textsc{Opt-WI}}\xspace}
\newcommand{\optsi}{{\normalfont \textsc{Opt-SI}}\xspace}
\newcommand{\uoptwi}{{\normalfont \textsc{UOpt-WI}}\xspace}
\newcommand{\uoptsi}{{\normalfont \textsc{UOpt-SI}}\xspace}
\newcommand{\exactwi}{{\normalfont \textsc{Exact-WI}}\xspace}
\newcommand{\exactsi}{{\normalfont \textsc{Exact-SI}}\xspace}
\newcommand{\QSAT}{{\normalfont \textsc{QSAT}\xspace}}
\newcommand{\MQSAT}{{\normalfont \textsc{MinQSAT}\xspace}}
\newcommand{\MWQSAT}{{\normalfont \textsc{Weighted MinQSAT}\xspace}}

\newcommand{\wFun}{\mthfun{w}}
\newcommand{\cFun}{\mthfun{c}}
\newcommand{\cost}{\mthfun{cost}}

\newcommand{\LTLlim}{\mthfun{LTL^{Lim}}\xspace}

\def\avg{{\sf avg}}

\def\src{{\sf src}}
\def\trg{{\sf trg}}
\def\IN{{\sf in}}
\def\OUT{{\sf out}}

\newcommand{\pay}{\mthfun{pay}}

\newcommand{\LP}{\mthfun{LP}}

\def\pspace{\mthfun{PSPACE}\xspace}
\def\fpspace{\mthfun{FPSPACE}\xspace}

\def\np{\mthfun{NP}\xspace}

\def\FP{\mthfun{FP}\xspace}

\def\DP{\mthfun{D^P}\xspace}
\def\DPTwo{\mthfun{D^P_{\normalfont \textrm{2}}}\xspace}

\newcommand{\sink}{\mthsym{sink}}
\newcommand{\source}{\mthsym{source}}

\newcommand{\setx}{\mathbf{x}}
\newcommand{\sety}{\mathbf{y}}
\newcommand{\vecx}{\vec{\mathbf{x}}}
\newcommand{\vecy}{\vec{\mathbf{y}}}
\newcommand{\term}{\mathbf{t}}
\newcommand{\T}{\mathbf{T}}
\newcommand{\SigmaPTwo}{\Sigma^{\mthfun{P}}_2}
\newcommand{\DeltaPTwo}{\Delta^{\mthfun{P}}_2}
\newcommand{\DeltaPThree}{\Delta^{\mthfun{P}}_3}

\usepackage{xargs}

\usepackage{xspace}

\usepackage{xstring}

\usepackage{boolexpr}

\usepackage[cmex10]{mathtools}
\usepackage{amssymb}
\usepackage{amsfonts}
\usepackage{latexsym}
\usepackage{textcomp}
\usepackage{pifont}

\newcommand{\argemp}[2]{\if&#1&\else#2\fi}

\newcommand{\argdef}[2]{\if&#1&#2\else#1\fi}

\newcommand{\argint}[3]{\if&#2&\else#1#2#3\fi}

\newcommand{\argext}[3]{\if&#1&#3\else#1\if&#3&\else#2#3\fi\fi}

\newcommandx{\mthfnt}[3][1=, 2=0]{{
	\IfStrEqCase{#1}
	{%
		{}%
		{#3}%
		{Name}%
		{%
			\IfStrEqCase{#2}
			{%
				{0}{\mathcal{#3}}%
				{1}{\mathscr{#3}}%
				{2}{\mathfrak{#3}}%
				{3}{\mathbb{#3}}%
			}
			[\ensuremath{\clubsuit}]%
		}%
		{Set}%
		{%
			\IfStrEqCase{#2}
			{%
				{0}{\mathrm{#3}}%
				{1}{\mathsf{#3}}%
				{2}{\mathbb{#3}}%
				{3}{\mathbf{#3}}%
			}
			[\ensuremath{\clubsuit}]%
		}%
		{Fun}%
		{%
			\IfStrEqCase{#2}
			{%
				{0}{\mathsf{#3}}%
				{1}{\mathrm{#3}}%
			}
			[\ensuremath{\clubsuit}]%
		}%
		{Rel}%
		{%
			\IfStrEqCase{#2}
			{%
				{0}{\mathit{#3}}%
				{1}{\mathtt{#3}}%
			}
			[\ensuremath{\clubsuit}]%
		}%
		{Sym}%
		{%
			\IfStrEqCase{#2}
			{%
				{0}{\mathtt{#3}}%
				{1}{\mathbf{#3}}%
			}
			[\ensuremath{\clubsuit}]%
		}%
		{Elm}%
		{\mathnormal{#3}}
	}
[\ensuremath{\clubsuit}]%
}}

\newcommand{\mthsub}[1]{\argemp{#1}{\ensuremath{_{\mathnormal{#1}}}}}

\newcommand{\mthsup}[1]{\argemp{#1}{\ensuremath{^{\mathnormal{#1}}}}}

\newcommandx{\mth}[5][1=, 2=0, 4=, 5=]{{\ensuremath{\mthfnt[#1][#2]{#3}\mthsub{#4}\mthsup{#5}}}}

\newcommandx{\mtharg}[6][1=, 2=0, 4=, 5=]{{\mth[#1][#2]{#3}[#4][#5]\ensuremath{\argint{(}{#6}{)}}}}

\newcommand{\mthempty}{\mth[][]}

\newcommand{\mthstyname}{0}
\newcommand{\mthname}[1][]{\mth[Name][\argdef{#1}{\mthstyname}]}

\newcommand{\mthstyset}{0}
\newcommand{\mthset}[1][]{\mth[Set][\argdef{#1}{\mthstyset}]}

\newcommand{\mthstyfun}{0}
\newcommand{\mthfun}[1][]{\mth[Fun][\argdef{#1}{\mthstyfun}]}

\newcommand{\mthstysym}{0}
\newcommand{\mthsym}[1][]{\mth[Sym][\argdef{#1}{\mthstysym}]}

\newcommand{\mthstyelm}{0}
\newcommand{\mthelm}[1][]{\mth[Elm][\argdef{#1}{\mthstyelm}]}

\newcommand{\tuple}[1]
{\ensuremath{\!\argint{\langle}{#1}{\rangle}}}

\usepackage[ruled,vlined,linesnumbered]{algorithm2e}

\usepackage{booktabs}
\usepackage[makeroom]{cancel}

\title{Equilibrium Design for Concurrent Games} 

\author{Julian Gutierrez}{Department of Computer Science, University of Oxford}{julian.gutierrez@cs.ox.ac.uk}{[orcid]}{}
\author{Muhammad Najib}{Department of Computer Science, University of Oxford}{mnajib@cs.ox.ac.uk}{[orcid]}{}
\author{Giuseppe Perelli}{Department of Computer Science, University of G\"{o}teborg}{giuseppe.perelli@gu.se}{[orcid]}{}
\author{Michael Wooldridge}{Department of Computer Science, University of Oxford}{michael.wooldridge@cs.ox.ac.uk}{[orcid]}{}

\authorrunning{J. Guiterrez, M. Najib, G. Perelli, M. Wooldridge}

\Copyright{Julian Guiterrez, Muhammad Najib, Giuseppe Perelli, Michael Wooldridge}



\category{}

\relatedversion{}

\supplement{}

\acknowledgements{Najib acknowledges the financial support of the Indonesia Endowment Fund for Education (LPDP), and Perelli the support of the project ``dSynMA'', funded by the ERC under the European Union's Horizon 2020 research and innovation programme (grant agreement No 772459).} 

\nolinenumbers 

\EventEditors{Wan Fokkink and Rob van Glabbeek}
\EventNoEds{2}
\EventLongTitle{30th International Conference on Concurrency Theory (CONCUR 2019)}
\EventShortTitle{CONCUR 2019}
\EventAcronym{CONCUR}
\EventYear{2019}
\EventDate{August 27--30, 2019}
\EventLocation{Amsterdam, the Netherlands}
\EventLogo{}
\SeriesVolume{140}
\ArticleNo{18}

\begin{document}

\maketitle

\begin{abstract}
In game theory, {\em mechanism design} is concerned with the design of
incentives so that a desired outcome of the game can be achieved. In this
paper, we study the design of incentives so that a desirable equilibrium is
obtained, for instance, an equilibrium satisfying a given temporal logic
property---a problem that we call {\em equilibrium design}. We base our study
on a framework where system specifications are represented as temporal logic
formulae, games as quantitative concurrent game structures, and players' goals
as mean-payoff objectives. In particular, we consider system specifications
given by LTL and GR(1) formulae, and show that implementing a mechanism to
ensure that a given temporal logic property is satisfied on some/every Nash
equilibrium of the game, whenever such a mechanism exists, can be done in
PSPACE for LTL properties and in NP/$\SigmaPTwo$ for GR(1) specifications. We
also study the complexity of various related decision and optimisation
problems, such as optimality and uniqueness of solutions, and show that the
complexities of all such problems lie within the polynomial hierarchy. As an
application, equilibrium design can be used as an alternative solution to the
rational synthesis and verification problems for concurrent games with
mean-payoff objectives whenever no solution exists, or as a technique to
repair, whenever possible, concurrent games with undesirable rational outcomes
(Nash equilibria) in an optimal way. 
\end{abstract}

\section{Introduction}
Over the past decade, there has been increasing interest in the use of
game-theoretic equilibrium concepts such as Nash equilibrium in the
analysis of concurrent and multi-agent systems (see,
{\em e.g.},~\cite{AKP18,AminofMMR16,BouyerBMU15,FismanKL10,GHPW17,GutierrezHW17-aij,KupfermanPV16}). 
This work views a concurrent system as a game, with system components
(agents) corresponding to players in the game, which are assumed to be
acting rationally in pursuit of their individual
preferences. Preferences may be specified by associating with each
player a temporal logic goal formula, which the player desires to see
satisfied, or by assuming that players receive rewards in each state
the system visits, and seek to maximise the average reward they
receive (the \emph{mean payoff}). A further possibility is to combine
goals and rewards: players primarily seek the satisfaction of their
goal, and only secondarily seek to maximise their mean payoff. The key
decision problems in such settings relate to what temporal logic
properties hold on computations of the system that may be generated by
players choosing strategies that form a game-theoretic (Nash)
equilibrium. These problems are typically computationally complex,
since they subsume temporal logic synthesis~\cite{PnueliR89}. If
players have \LTL goals, for example, then checking whether an \LTL
formula holds on some Nash equilibrium path in a concurrent game is
2EXPTIME-complete~\cite{FismanKL10,GutierrezHW15,GutierrezHW17-aij}, 
rather than only PSPACE-complete as it is the case for model checking, 
certainly a computational barrier for the practical analysis and automated 
verification of reactive, concurrent, and multi-agent systems modelled as 
multi-player games. 

Within this game-theoretic reasoning framework, a key issue 
is that individually rational choices can
cause outcomes that are highly undesirable, and concurrent games also
fall prey to this problem. This has motivated the development of
techniques for modifying games, in order to avoid bad equilibria, or
to facilitate good equilibria. \emph{Mechanism design} is the problem of
designing a game such that, if players behave rationally, then a
desired outcome will be obtained~\cite{OR94}. Taxation and subsidy
schemes are probably the most important class of techniques used in
mechanism design. They work by levying taxes on certain actions (or
providing subsidies), thereby incentivising
players away from some outcomes towards others. The present paper
studies the design of subsidy schemes (incentives) for concurrent games, 
so that a desired outcome (a Nash equilibrium in the game) can be obtained---a 
problem that we call {\em Equilibrium design}. 
We model agents as synchronously executing concurrent processes, with 
each agent receiving an integer payoff for every state the overall system visits;
the overall payoff an agent receives over an infinite computation path
is then defined to be the mean payoff over this path. While agents (naturally) seek to maximise their individual mean payoff, the designer of the subsidy scheme wishes to see some temporal logic formula satisfied, either 
on some or on every Nash equilibrium of the game. 

With this model, we assume that the designer -- an external principal -- has a 
finite budget that is available for making subsidies, and this budget can be
allocated across agent/state pairs. By allocating this budget
appropriately, the principal can incentivise players away from some
states and towards others. Since the principal has some temporal
logic goal formula, it desires to allocate subsidies so that players
are rationally incentivised to choose strategies so that the
principal's temporal logic goal formula is satisfied in the path that would result from
executing the strategies.  For this general problem,
following~\cite{WEKL13}, we identify two variants of the principal's
mechanism design problem, which we refer to as \textsc{Weak
	Implementation} and \textsc{Strong Implementation}. In the
\textsc{Weak} variant, we ask whether the principal can allocate the
budget so that the goal is achieved on \emph{some} computation path that would be
generated by Nash equilibrium strategies in the resulting system; in
the \textsc{Strong} variation, we ask whether the principal can
allocate the budget so that the resulting system has at least one Nash
equilibrium, and moreover the temporal logic goal is satisfied on {\em all} paths that
could be generated by Nash equilibrium strategies. For these two
problems, we consider goals specified by \LTL formulae or
\GRone~formulae \cite{BJPPS12}, give algorithms for each case, and
classify the complexity of the problem. 
While \LTL is a natural language for the specification of properties of concurrent 
and multi-agent systems, \GRone is an \LTL fragment that can be used to 
easily express several prefix-independent properties of computation paths of reactive systems, 
such as $\omega$-regular properties often used in automated formal verification. 
We then go on to study
variations of these two problems, for example considering 
{\em optimality} and {\em uniqueness} of solutions, and show that the
complexities of all such problems lie within the polynomial hierarchy, thus
making them potentially amenable to efficient practical implementations.  
Table~\ref{tab:results} summarises the main computational complexity results in the paper. 


\begin{table*}[!ht]
	\begin{center}
		\def\arraystretch{1.5}
		\begin{tabular}{c c c}
			\toprule
			& \LTL Spec. & \GRone Spec. \\
			\hline
			\weak & \pspace-complete (Thm. \ref{thm:weak-ltl}) & \np-complete (Thm. \ref{thm:weak-gr1}) \\
			\strong & \pspace-complete (Cor. \ref{cor:strong-ltl}) & $ \SigmaPTwo $-complete (Thm. \ref{thm:strong-gr1}) \\
			\optwi & \fpspace-complete (Thm. \ref{thm:optwi-ltl}) & $ \FP^{\np} $-complete (Thm. \ref{thm:optwi-gr1}) \\
			\optsi & \fpspace-complete (Thm. \ref{thm:optsi-ltl}) & $ \FP^{\SigmaPTwo} $-complete (Thm. \ref{thm:optsi-gr1}) \\
			\exactwi & \pspace-complete (Cor. \ref{cor:exactwi-ltl}) & $ \DP $-complete (Cor. \ref{cor:exactwi-gr1}) \\
			\exactsi & \pspace-complete (Cor. \ref{cor:exactsi-ltl}) & $ \DPTwo $-complete (Cor. \ref{cor:exactsi-gr1}) \\
			\uoptwi & \pspace-complete (Cor. \ref{cor:uoptwi-ltl}) & $ \DeltaPTwo $-complete (Cor. \ref{cor:uoptwi-gr1}) \\
			\uoptsi & \pspace-complete (Cor. \ref{cor:uoptsi-ltl}) & $ \DeltaPThree $-complete (Cor. \ref{cor:uoptsi-gr1}) \\
			\bottomrule
		\end{tabular}
	\end{center}
	\caption{Summary of main complexity results.}
	\label{tab:results}
\end{table*}

\section{Preliminaries}\label{sec:prelims}
\noindent \textbf{Linear Temporal Logic.}
\LTL~\cite{pnueli:77a} extends classical propositional logic with two
operators, $\ltlnext$ (``next'') and $\until$ (``until''), that can be used to express properties of paths.  The syntax of \LTL is defined
with respect to a set $\AP$ of atomic propositions as follows:
$$ \phi ::= 
\mathop\top \mid
p \mid
\neg \phi \mid\phi \vee \phi \mid
\ltlnext \phi \mid
\phi \until \phi
$$
where $p \in \AP$.
As commonly found in the \LTL literature, we  use of the following abbreviations: $\phi_1 \wedge \phi_2 \equiv \neg (\neg \phi_1 \vee \neg \phi_2)$, $\phi_1 \to \phi_2 \equiv \neg \phi_1 \vee \phi_2$, $\sometime \phi \equiv \mathop\top \until \phi$, and $\always \phi \equiv \neg \sometime \neg \phi$.\

We interpret formulae of \LTL with respect to pairs $(\alpha,t)$, where $\alpha \in (2^{\AP})^\omega$ is an infinite sequence of atomic proposition evaluations that indicates which propositional variables are true in every time point and $t \in \Nat$ is a
temporal index into $\alpha$.
Formally, the semantics of \LTL formulae is given by the following rules:
$$
\begin{array}{lcl}
(\alpha,t)\models\mathop\top	\\
(\alpha,t)\models p 				&\text{ iff }&	p\in \alpha_t\\
(\alpha,t)\models\neg \phi			&\text{ iff }&   \text{it is not the case that $(\alpha,t) \models \phi$}\\
(\alpha,t)\models\phi \vee \psi		&\text{ iff }&	\text{$(\alpha,t) \models \phi$  or $(\alpha,t) \models \psi$}\\
(\alpha,t)\models\ltlnext\phi			&\text{ iff }&	\text{$(\alpha,t+1) \models \phi$}\\
(\alpha,t)\models\phi\until\psi	&\text{ iff }&   \text{for some $t' \geq t: \ \big((\alpha,t') \models \psi$  and }\\
&&\quad\text{for all $t \leq t'' < t': \ (\alpha,t'') \models \phi \big)$.}\\
\end{array}
$$
If $(\alpha,0)\models\phi$, we write $\alpha\models\phi$ and say that
\emph{$\alpha$ satisfies~$\phi$}.

\vspace*{4pt} 
\noindent \textbf{General Reactivity of rank 1.}
The language of \emph{General Reactivity of rank 1}, denoted $\GRone$, is the fragment of \LTL given by  formulae written in the following form~\cite{BJPPS12}:
$$
(\always \sometime \psi_1 \wedge \ldots \wedge \always \sometime \psi_m) \to (\always \sometime \phi_1 \wedge \ldots \wedge \always \sometime \phi_n)
\text{,}
$$
where each subformula $\psi_i$ and $\phi_i$ is a Boolean combination of atomic propositions.

\vspace{4pt}
\noindent \textbf{Mean-Payoff.}
For a sequence $r \in \mathbb{R}^\omega$, let $\MP(r)$ be
the \emph{mean-payoff} value of $r$, that is, 
$$ \MP(r) = \lim \inf_{n \to \infty} \avg_n(r) $$
where, for $n \in \mathbb{N}\setminus\{0\}$, we define
$\avg_n(r) = \frac{1}{n}\sum_{j=0}^{n-1} r_j$, with $r_j$ the $(j\!+\!1)$th element of $r$. 

\vspace*{4pt} \noindent \textbf{Arenas.}
An \emph{arena} is a tuple
$ A = \tuple{\Ag,  \Ac, \St, s_0, \trnFun, \labFun} $ 
where $\Ag$, $\Ac$, and $\St$ are finite non-empty sets of \emph{players} (write $N = \card{\Ag}$), \emph{actions}, and \emph{states}, respectively; if needed, we write $ \Ac_i(s) $, to denote the set of actions available to player $ i $ at $ s $; $s_0 \in \St$ is the \emph{initial state}; $\trnFun : \St \times \AcProf \rightarrow \St$ is a \emph{transition function} mapping each pair consisting of a state $s \in \St$ and an \emph{action profile} $\jact \in \AcProf = \Ac^{\Ag}$, one for each player, to a successor state; and $\labFun: \St \to 2^{\AP}$ is a labelling function, mapping every state to a subset of \emph{atomic propositions}.

We sometimes call an action profile $\jact = (\act_{1}, \dots, \act_{n}) \in \AcProf$ a \emph{decision}, and denote $\act_i$ the action taken by player $i$.
We also consider \emph{partial} decisions.
For a set of players $C \subseteq \Ag$ and action profile $\jact$, we let $\jact_{C}$ and $\jact_{-C}$ be two tuples of actions, respectively, one for all players in $C$ and one for all players in $\Ag \setminus C$.
We also write $\jact_{i}$ for $\jact_{\{i\}}$ and $ \jact_{-i} $ for $ \jact_{\Ag \setminus \{i\}} $.
For two decisions $\jact$ and $\jact'$, we write $(\jact_{C}, \jact_{-C}')$ to denote the decision where the actions for players in $ C $ are taken from $\jact$ and the actions for players in $ \Ag \setminus C $ are taken from $\jact'$.

A \emph{path} $\pi = (s_0, \jact^0), (s_1, \jact^1) \cdots$ is an infinite sequence in $(\St \times \AcProf)^{\omega}$ such that $\trnFun(s_k, \jact^k) = s_{k + 1}$ for all $k$.
%
Paths are generated in the arena by each player~$i$ selecting a {\em
	strategy} $\strElm_i$ that will define how to make choices over
time.  We model strategies as finite state machines with output.
Formally, for arena $A$, a strategy
$\strElm_{i} = (Q_{i}, q_{i}^{0}, \delta_i, \tau_i) $ for player $i$
is a finite state machine with output (a transducer), where $Q_{i}$ is
a finite and non-empty set of \emph{internal states}, $ q_{i}^{0} $ is
the \emph{initial state},
$\delta_i: Q_{i} \times \AcProf \rightarrow Q_{i} $ is a deterministic
\emph{internal transition function}, andlet me
$\tau_i: Q_{i} \rightarrow \Ac_i$ an \emph{action function}. Let $\StrSet_i$ be the set of strategies for player $i$. Note that this definition implies that strategies have \textit{perfect information}\footnote{Mean-payoff games with imperfect information are generally undecidable \cite{DDGRT2010}.} and finite memory (although we impose no bounds on memory size).

A \emph{strategy profile} $\strpElm = (\strElm_1, \dots, \strElm_n)$ is a vector of strategies, one for each player.
As with actions, $\strpElm_{i}$ denotes the strategy assigned to player $i$ in profile $\strpElm$.
Moreover, by $(\strpElm_{B}, \strpElm'_{C})$ we denote the combination
of profiles where players in disjoint $B$ and $C$ are assigned their corresponding strategies in $\strpElm$ and $\strpElm'$, respectively.
Once a state $s$ and profile $\strpElm$ are fixed, the game has an \emph{outcome}, a path in $A$, denoted by $\pi(\strpElm, s)$. 
Because strategies are deterministic, $\pi(\strpElm, s)$ is the unique path induced by $\strpElm$, that is, the sequence $s_0, s_1, s_2, \ldots$ such that 
\begin{itemize}
	\item $s_{k + 1} = \trnFun (s_k, (\tau_1(q^k_1), \ldots, \tau_n(q^k_n)))$, and 
	\item $q^{k + 1}_i = \delta_i(s^k_i, (\tau_1(q^k_1), \ldots, \tau_n(q^k_n)))$, for all $k \geq 0$. 
\end{itemize}
Furthermore, we simply write $ \pi(\strpElm) $ for $ \pi(\strpElm,s_0) $.

Arenas define the dynamic structure of games, but lack a central aspect of a game: preferences, which give games their strategic structure.
A \emph{multi-player game} is obtained from an arena $A$ by
associating each player with a goal.
We consider multi-player games with $\MP$ goals.
%
A multi-player \MP game is a tuple
$\Game = \tuple{A, (\wFun_{i})_{i \in \Ag}}$, where $A$ is an
arena and $\wFun_{i}: \St \to \SetZ$ is a function mapping, for every player~$i$, every state
of the arena into an integer number.  
%
In any game with arena $A$, a path $\pi$ in $A$ induces a sequence $\lambda(\pi) = \lambda(s_0) \lambda(s_1) \cdots$ of sets of atomic propositions; if, in addition, $A$ is the arena of an \MP game, then, for each player~$i$, the sequence $\wFun_i(\pi) = \wFun_i(s_0) \wFun_i(s_1) \cdots$ of weights is also induced. 
Unless stated otherwise, for a game $ \Game $ and a path $\pi$ in it, the payoff of player~$i$ is $\pay_i(\pi) = \MP(\wFun_{i}(\pi))$.

\vspace*{4pt} \noindent \textbf{Nash equilibrium.}
Using payoff functions, we can define the game-theoretic concept of Nash equilibrium~\cite{OR94}. 
For a multi-player game $\Game$, a strategy profile
$\strpElm$ is a \emph{Nash equilibrium} of~$\Game$ if, for every player~$i$ and strategy $\strElm'_i$ for player~$i$, we have
$$
\pay_i(\pi(\strpElm))	\geq	\pay_i(\pi((\strpElm_{-i},\strElm'_i))) \ . 
$$
Let $\NE(\Game)$ be the set of Nash equilibria of~$\Game$.

\section{From Mechanism Design to Equilibrium Design}
\label{sec:eqdesign}
We now describe the two main problems that are our focus of study. As discussed in the introduction, such problems are closely related to the well-known problem of {\em mechanism design} in game theory. 
Consider a system populated by agents \Ag,
where each agent $ i \in \Ag $ wants to maximise its payoff $ \pay_i(\cdot) $.
As in a mechanism design problem, we assume there is an external \textit{principal} who has a goal $ \phi $ that it
wants the system to satisfy, and to this end, wants to incentivise the agents to act collectively and rationally so as to bring about $ \phi $. In our
model, incentives are given by \textit{subsidy schemes} and goals by temporal logic formulae. 

\vspace{4pt}
\noindent \textbf{Subsidy Schemes:} A subsidy scheme defines additional
imposed rewards over those given by the weight function $ \wFun $. 
While the weight function $ \wFun$ is fixed for any
given game, the principal is assumed to be at liberty to define a subsidy 
scheme as they see fit. Since agents will seek to maximise their overall rewards,
the principal can incentivise agents away from performing visiting some states and
towards visiting others; if the principal designs the subsidy scheme
correctly, the agents are incentivised to choose a strategy profile $ \vec{\sigma} $
such that $ \pi(\strpElm) \models \phi $.
Formally, we model a subsidy scheme as a function $ \kappa: \Ag \to \St \to \mathbb{N} $, where the intended 
interpretation is that $ \kappa(i)(s) $ is the subsidy in the form of a natural number $ k \in \mathbb{N}$ that would be imposed on player $ i $ if such a player visits state $ s \in \St $. For instance, if we have $ \wFun_{i}(s) = 1 $ and $ \kappa(i)(s) = 2 $, then player~$i$ gets $1 + 2 = 3 $ for visiting such a state. For simplicity, hereafter we write $ \kappa_i(s) $ instead of $ \kappa(i)(s)$ for the subsidy for player~$i$. 

Notice that having an unlimited fund for a subsidy scheme would make some problems trivial, as the principal can always incentivise players to satisfy $ \phi $ (provided that there is a path in $ A $ satisfying $ \phi $). A natural and more interesting setting is that the principal is given a constraint in the form of \textit{budget} $\beta\in\mathbb{N}$. The principal then can only spend within the budget limit. To make this clearer, we first define the \textit{cost} of a subsidy scheme $ \kappa $ as follows.

\begin{definition}\label{def:cost}
	Given a game $ \Game $ and subsidy scheme $ \kappa $, we let $ \cost(\kappa) = \sum_{i \in \Ag} \sum_{s \in \St} \kappa_i(s) $.
\end{definition}

We say that a subsidy scheme $ \kappa $ is \textit{admissible} if it does not exceed the budget~$\beta$,  that is, if $ \cost(\kappa) \leq \beta $. Let $ \K(\Game,\beta) $ denote the set of admissible subsidy schemes over $ \Game $ given budget $ \beta \in \mathbb{N} $. Thus we know that for each $ \kappa \in \K(\Game,\beta)$ we have $\cost(\kappa) \leq \beta $. We write $ (\Game,\kappa) $ to denote the resulting game after the application of subsidy scheme $ \kappa $ on game $ \Game $. Formally, we define the application of some subsidy scheme on a game as follows.

\begin{definition}\label{def:apply}
	Given a game $ \Game = \tuple{A,(\wFun_{i})_{i \in \Ag}} $ and an admissible subsidy scheme $ \kappa $, we define $ (\Game,\kappa) = \tuple{A, (\wFun'_{i})_{i \in \Ag}} $, where $ \wFun'_{i}(s) = \wFun_{i}(s) + \kappa_i(s)$, for each $ i \in \Ag $ and $ s \in \St $.
\end{definition}

We now come to the main question(s) that we consider in the remainder of the paper. We ask whether the principal can find a subsidy scheme that will incentivise players to collectively choose a rational outcome (a Nash equilibrium) that satisfies its temporal logic goal $ \phi $. We call this problem {\em equilibrium design}. Following~\cite{WEKL13}, we define two variants of this problem, a {\em weak} and a {\em strong} implementation of the equilibrium design problem. The formal definition of the problems and the analysis of their respective computational complexity are presented in the next sections.



\section{Equilibrium Design: Weak Implementation}
In this section, we study the weak implementation of the equilibrium design problem, a logic-based computational variant of the principal's mechanism design problem in game theory. We assume that the principal has full knowledge of the game $ \Game $ under consideration, that is, the principal uses all the information available of $ \Game $ to find the appropriate subsidy scheme, if such a scheme exists. We now formally define the weak variant of the implementation problem, and study its respective computational complexity, first with respect to goals (specifications) given by \LTL formulae and then with respect to \GRone formulae. 


Let $ \WI(\Game,\phi,\beta) $ denote the set of subsidy schemes over $ \Game $ given budget $ \beta $ that satisfy a formula $ \phi $ in at least one path $ \pi $ generated by $ \strpElm \in \NE(\Game) $. Formally
$$ \WI(\Game,\phi,\beta) = \{ \kappa \in \K(\Game,\beta) : \exists \strpElm \in \NE(\Game,\kappa)~\textrm{s.t.}~\pi(\strpElm) \models \phi \}. $$

\begin{definition}[\weak]
	Given a game $\Game$, formula $\varphi$, and budget $ \beta $:
	\begin{center}
		Is it the case that $ \WI(\Game,\phi,\beta) \neq \varnothing $?
	\end{center}
\end{definition}

In order to solve \weak, we first characterise the Nash equilibria of a multi-player concurrent game in terms of punishment strategies. To do this in our setting, we recall the notion of secure values for mean-payoff games~\cite{UW11}.

For a player $i$ and a state $s \in \St$, by $\pun_i(s)$ we denote the
punishment value of $i$ over $s$, that is, the maximum payoff that $i$
can achieve from $s$, when all other players behave adversarially.
Such a value can be computed by considering the corresponding
two-player zero-sum mean-payoff game~\cite{ZP96}.  Thus, it is in
$\np \cap \mathsf{co}\np$, and note that both player $i$ and coalition
$\Ag \setminus \{i\}$ can achieve the optimal value of the game using
{\em memoryless} strategies.
Then, for a player $i$ and a value $z \in \SetR$, a pair $(s, \jact)$ is $z$-secure for player~$i$ if $\pun_i(\trnFun(s, (\jact_{-i}, \act'_i))) \leq z$ for every $\act'_i \in \Ac$. Write $\pun_{i}(\Game)$ for the set of punishment values for player~$i$ in $\Game$.

\begin{theorem}
	\label{thm:pthfinding}
	For every \MP game $\Game$ and ultimately periodic path $\pi = (s_0, \jact_{0}), (s_1, \jact^{1}), \ldots $, the following are equivalent:
	
	\begin{enumerate}
		\item 
		There is $\strpElm \in \NE(\Game)$ such that $\pi = \pi(\strpElm, s_0)$;
		
		\item 
		There exists $ {z} \in \SetR^{\Ag}$, where $z_{i} \in \pun_{i}(\Game)$ such that, for every $i \in \Ag$
		
		\begin{enumerate}
			\item 
			for all $k \in \SetN$, the pair $(s_k, \jact^k)$ is $z_i$-secure for $i$, and 
			
			\item 
			$z_i \leq \pay_i(\pi)$.
		\end{enumerate}
	\end{enumerate}
	
\end{theorem}

%

The characterisation of Nash Equilibria provided in Theorem~\ref{thm:pthfinding} will allow us to turn the \weak problem into a {\em path finding} problem over $(\Game,\kappa)$. 
On the other hand, with respect to the budget $\beta$ that the principal has at its disposal, the definition of subsidy scheme function $\kappa$ implies that the size of $ \K(\Game,\beta) $ is bounded, and particularly, it is bounded by $\beta$ and the number of agents and states in the game $\Game$, in the following way. 

\begin{proposition}\label{prop:kbound}
	Given a game $ \Game $ with $\card{N}$ players and $ \card{\St} $ states and budget $\beta$, it holds that 
	$$ \card{\K(\Game,\beta)} = \frac{\beta + 1}{m} \binom{\beta + m}{\beta + 1}\text{,}$$
	with $m = \card{N \times \St}$ being the number of pairs of possible agents and states.
\end{proposition}

From Proposition~\ref{prop:kbound} we derive that the number of possible subsidy schemes is {\em polynomial} in the budget $\beta$ and singly {\em exponential} in both the number of agents and states in the game.
At this point, solving \weak can be done with the following procedure:

\begin{enumerate}
	\item \label{proc:guesses}
	Guess: 
	\begin{itemize}
	\item a subsidy scheme $ \kappa \in \K(\Game,\beta) $,
	\item a state $ s \in \St $ for every player $ i \in \Ag $, and
	\item punishment memoryless strategies $ (\strpElm_{-1},\dots,\strpElm_{-n}) $ for all players $ i \in \Ag $;	
	\end{itemize}

	\item \label{proc:compute-gk}
	Compute $ (\Game,\kappa) $;
	
	\item \label{proc:compute-z}
	Compute $ z \in \SetR^{\Ag} $;
	
	\item \label{proc:compute-gz}
	Compute the game $(\Game,\kappa){[z]}$ by removing the states $s$ such that $\pun_i(s) \leq z_i$ for some player $i$ and the transitions $(s, \jact_{-i})$ that are not $z_i$ secure for player $i$;
	
	\item \label{proc:find-path}
	Check whether there exists an ultimately periodic path $ \pi $ in $ (\Game,\kappa){[z]}$ such that $ \pi \models \phi $ and $ z_i \leq \pay_i(\pi) $ for every player $ i \in \Ag $.
\end{enumerate}

Since the set $ \K(\Game,\beta) $ is finitely bounded (Proposition \ref{prop:kbound}), and punishment strategies only need to be memoryless, thus also finitely bounded, clearly step~\ref{proc:guesses} can be guessed nondeterministically. Moreover, each of the guessed elements is of polynomial size, thus this step can be done (deterministically) in polynomial space. 
Step~\ref{proc:compute-gk} clearly can be done in polynomial time.
Step~\ref{proc:compute-z} can also be done in polynomial time since, given $ (\strpElm_{-1},\dots,\strpElm_{-n}) $, we can compute $ z $ solving $|\Ag|$ one-player mean-payoff games, one for each player $ i $~\cite[Thm.~6]{ZP96}. For step~\ref{proc:find-path}, we will use Theorem~\ref{thm:pthfinding} and consider two cases, one for \LTL specifications and one for \GRone specifications. Firstly, for \LTL specifications, consider the formula
$$
\phi_{\WI} := \phi \wedge \bigwedge_{i \in \Ag} (\MP(i) \geq z_i)
$$
written in $\LTLlim$~\cite{BCHK14}, an extension of \LTL where statements about mean-payoff values over a given weighted arena can be made.%
\footnote{The formal semantics of $\LTLlim$ can be found in ~\cite{BCHK14}. We prefer to give only an informal description here.}
%
The semantics of the temporal operators of $\LTLlim$ is just like the one for \LTL over infinite computation paths $\pi = s_0,s_1,s_3.\ldots$. On the other hand, the meaning of $\MP(i) \geq z_i$ is simply that such an atomic formula is true if, and only if, the mean-payoff value of $\pi$ with respect to player~$i$ is greater or equal to $z_i$, a constant real value; that is, $\MP(i) \geq z_i$ is true in $\pi$ if and only if $\pay_i(\pi) = \MP(\wFun_{i}(\pi))$ is greater or equal than constant value $z_i$. 
Formula $ \phi_{\WI} $ corresponds exactly to $ 2(b) $ in Theorem \ref{thm:pthfinding}. Furthermore, since every path in $ (\Game,\kappa){[z]} $ satisfies condition $ 2(a) $ of Theorem \ref{thm:pthfinding}, every computation path of $ (\Game,\kappa){[z]} $ that satisfies $ \phi_{\WI} $ is a witness to the \weak problem.

\begin{theorem}\label{thm:weak-ltl}
	\weak with \LTL specifications is \pspace-complete.
\end{theorem}

\begin{proof}
Membership follows from the procedure above and the fact that model checking for $\LTLlim$ is \pspace-complete~\cite{BCHK14}. Hardness follows from the fact that \LTL model checking is a special case of \weak. For instance, consider the case in which all weights for all players are set to the same value, say 0, and the principal has budget $ \beta = 0 $.
\end{proof}

\noindent\textbf{Case with $ \GRone $ specifications.} 
One of the main bottlenecks of our procedure to solve \weak lies in step~\ref{proc:find-path}, where we solve an $ \LTLlim $ model checking problem. To reduce the complexity of our decision procedure, we consider \weak with the specification $ \phi $ expressed in the $ \GRone $ sublanguage of \LTL. With this specification language, the path finding problem can be solved without model-checking the $ \LTLlim $ formula given before. In order to do this, we can define a linear program (\LP) such that the \LP~has a solution if and only if~$ \WI(\Game,\phi,\beta) \neq \varnothing $. 
From our previous procedure, observe that step~\ref{proc:guesses} can be done nondeterministically in polynomial time, and steps~\ref{proc:compute-gk}--\ref{proc:compute-gz} can be done (deterministically) in polynomial time. Furthermore, using \LP, we also can check step~\ref{proc:find-path} deterministically in polynomial time.
For the lower-bound, we use \cite{UW11} and note that if $\phi=\top$ and $ \beta = 0 $, then the problem reduces to checking whether the underlying \MP game has a Nash equilibrium. Based on the above observations, we have the following result. 

\begin{theorem}\label{thm:weak-gr1}
	\weak with $ \GRone $ specifications is \np-complete.
\end{theorem}

\begin{proof}[Proof sketch]
For the upper bound, we define an \LP~of size polynomial in $(\Game,\kappa)$ having a solution if and only if there is an ultimately periodic path $ \pi $ such that $z_i \leq \pay_i(\pi) $ and satisfies the \GRone specification. 
	Recall that $\phi$ has the following form
	$$
	\phi = \bigwedge_{l = 1}^{m} \always \sometime \psi_{l} \to \bigwedge_{r = 1}^{n} \always \sometime \theta_{r}\text{,}
	$$
	and let $V(\psi_{l})$ and $V(\theta_r)$ be the subset of states in $(\Game,\kappa)$ that satisfy the Boolean combinations $\psi_{l}$ and $\theta_{r}$, respectively. Property $\phi$ is satisfied on $\pi$ if, and only if, either $\pi$ visits every state in $V(\theta_r)$ infinitely often or some of the states in $V(\psi_{l})$ only a finite number of times.
	For the game $(\Game,\kappa){[z]}$, let $ W = (V, E, (\wFun_a)_{a \in \Ag})$ be the underlying multi-weighted graph, and for every edge $e\in E$ introduce a variable $x_e$. Informally, the value of $x_e$ is the number of times that $e$ is used on a cycle. Formally, let $\src(e) = \{v \in V : \exists w\, e = (v,w) \in E\}$; $\trg(e) = \{v \in V : \exists w\, e = (w,v) \in E\}$; $\OUT(v) = \{e \in E : \src(e) = v\}$; and $\IN(v) = \{e \in E : \trg(e) = v\}$.
	Now, consider $\psi_{l}$ for some $1 \leq l \leq m$, and define the following linear program $\LP(\psi_{l})$: 
	\begin{enumerate}
		\item[Eq1:]
		$x_e \geq 0$ for each edge $e$ 
		--- a basic consistency criterion;
		
		\item[Eq2:]
		$\Sigma_{e \in E} x_e \geq 1$ 
		--- at least one edge is chosen;
		
		\item[Eq3:]
		for each $a \in \Ag$, $\Sigma_{e \in E} \wFun_a(\src(e)) x_e \geq 0$ --- total sum of any solution is non-negative;
		
		\item[Eq4:]
		$\Sigma_{\src(e) \cap V(\psi_{l}) \neq \emptyset} x_e = 0$ --- no state in $V(\psi_{l})$ is in the cycle associated with the solution;
		
		\item[Eq5:]
		for each $v \in V$, $\Sigma_{e \in \OUT(v)} x_e = \Sigma_{e \in \IN(v)} x_e$  --- this condition says that the number of times one enters a vertex is equal to the number of times one leaves that vertex.	
	\end{enumerate}
	$\LP(\psi_{l})$ has a solution if and only if there is a path $\pi$ in $\Game$ such that $z_i \leq \pay_i(\pi) $ for every player $i$ and visits $V(\psi_{l})$ only {\em finitely many times}.
	Consider now the linear program $\LP(\theta_{1}, \ldots, \theta_{n})$ defined as follows. Eq1--Eq3 as well as Eq5 are as in $\LP(\psi_{l})$, and:
	
	\begin{enumerate}		
	
		\item[Eq4:]
		for all $1 \leq r \leq n$, $\Sigma_{\src(e) \cap V(\theta_{r}) \neq \emptyset} x_e \geq 1$ ---  this condition says that, for every $V(\theta_{r})$, at least one state in $V(\theta_{r})$ is in the cycle associated with the solution of the linear program. 
		
	\end{enumerate}
	In this case, $\LP(\theta_{1}, \ldots, \theta_{n})$  has a solution if and only if there exists a path $\pi$ such that $z_i \leq \pay_i(\pi) $ for every player $i$ and visits every $V(\theta_{r})$ {\em infinitely many times}.
	Since the constructions above are polynomial in the size of both $(\Game,\kappa)$ and $\phi$, we can conclude it is possible to check in \np the statement that there is a path $\pi$ satisfying $\varphi$ such that $z_i \leq \pay_i(\pi) $ for every player~$i$ in the game if and only if one of the two linear programs defined above has a solution.
	For the lower-bound, we use \cite{UW11} as discussed before. 
\end{proof}

We now turn our attention to the strong implementation of the equilibrium design problem. As in this section, we first consider \LTL specifications and then \GRone specifications. 

\section{Equilibrium Design: Strong Implementation}
Although the principal may find $ \WI(\Game,\phi,\beta) \neq \varnothing $ to be good news, it might not be good enough. It could be that even though there is a desirable Nash equilibrium, the others might be undesirable. This motivates us to consider the \textit{strong implementation} variant of equilibrium design. Intuitively, in a strong implementation, we require that \textit{every} Nash equilibrium outcome satisfies the specification~$ \phi $, for a {\em non-empty} set of outcomes. Then, let $ \SI(\Game,\phi,\beta) $ denote the set of subsidy schemes $ \kappa $ given budget $ \beta $ over $ \Game $ such that:
\begin{enumerate}
	\item $ (\Game,\kappa) $ has at least one Nash equilibrium outcome,
	\item every Nash equilibrium outcome of $ (\Game,\kappa) $ satisfies $ \phi $.
\end{enumerate}

Formally we define it as follows:
$$
\SI(\Game,\phi,\beta) = \{ \kappa \in \mathcal{K}(\Game,\beta) : \NE(\Game,\kappa) \neq \varnothing \wedge \forall \strpElm \in \NE(\Game,\kappa)~\textrm{s.t.}~\pi(\strpElm) \models \phi \}.
$$
This gives us the following decision problem:
\begin{definition}[\strong]
	Given a game $\Game$, formula $\varphi$, and budget $ \beta $:
	\begin{center}
		Is it the case that $ \SI(\Game,\phi,\beta) \neq \varnothing $?
	\end{center}
\end{definition}

\strong can be solved with a 5-step procedure where the first four steps are as in \weak, and the last step (step~5) is as follows:%
%
%
%
%
%
%
%
%
\begin{enumerate}
	\item[{\sf 5}] \label{proc:find-path-s}
	Check whether:
	\begin{enumerate}
		\item there is no ultimately periodic path $ \pi $ in $ (\Game,\kappa){[z]}$ such that $ z_i \leq \pay_i(\pi) $ for each $ i \in \Ag $;
		\item there is an ultimately periodic path $ \pi $ in $ (\Game,\kappa){[z]}$ such that $ \pi \models \lnot\phi $ and $ z_i \leq \pay_i(\pi) $, for each $ i \in \Ag $.
	\end{enumerate}
\end{enumerate}

 For step~\ref{proc:find-path-s}, observe that a positive answer to \ref{proc:find-path-s}(a) or \ref{proc:find-path-s}(b) is a counterexample to $ \kappa \in \SI(\Game,\phi,\beta) $. Then, to carry out this procedure for the \strong problem with \LTL specifications, consider the following \LTLlim formulae:

\begin{ceqn}
	\begin{align*}
	\phi_{\exists} &= \bigwedge_{i \in \Ag} (\MP(i) \geq z_i);\\
	\phi_{\forall} &= \phi_{\exists} \rightarrow \phi.
	\end{align*}
\end{ceqn}

Notice that the expression $ \NE(\Game,\kappa) \neq \varnothing $ can be expressed as ``there exists a path $ \pi $ in $ \Game $ that satisfies formula $ \phi_{\exists} $''. On the other hand, the expression $ \forall \strpElm \in \NE(\Game,\kappa)~\textrm{such that}~\pi(\strpElm) \models \phi $ can be expressed as ``for every path $ \pi $ in $ \Game $, if $\pi$ satisfies formula $ \phi_{\exists} $, then $\pi$ also satisfies formula $ \phi $''. Thus, using these two formulae, we obtain the following result. 

\begin{corollary}\label{cor:strong-ltl}
	\strong with \LTL specifications is \pspace-complete.
\end{corollary}

\begin{proof}
	Membership follows from the fact that step~\ref{proc:find-path-s}(a) can be solved by existential \LTLlim model checking, whereas step~\ref{proc:find-path-s}(b) by universal \LTLlim model checking---both clearly in \pspace by Savitch's theorem. Hardness is similar to the construction in Theorem \ref{thm:weak-ltl}.
\end{proof}

\noindent\textbf{Case with $ \GRone $ specifications.}
Notice that the first part, {\em i.e.}, $ \NE(\Game,\kappa) \neq \varnothing $ can be solved in \np \cite{UW11}. For the second part, observe that $$ \forall \strpElm \in \NE(\Game,\kappa)~\textrm{such that}~\pi(\strpElm) \models \phi $$ is equivalent to $$ \neg\exists \strpElm \in \NE(\Game,\kappa)~\text{such that}~\pi(\strpElm) \models \lnot \phi. $$

Thus we have
$$
\lnot\varphi = \bigwedge_{l = 1}^{m} \always \sometime \psi_{l} \wedge \lnot \big( \bigwedge_{r = 1}^{n} \always \sometime \theta_{r} \big)\text{.}
$$
To check this, we modify the \LP~in Theorem \ref{thm:weak-gr1}. Specifically, we modify Eq4 in $\LP(\theta_{1}, \ldots, \theta_{n})$ to encode the $ \theta $-part of $ \lnot \phi $. Thus, we have the following equation in $\LP'(\theta_{1}, \ldots, \theta_{n})$:

\begin{enumerate}
	
	\item[Eq4:]
	there exists $ r $, $1 \leq r \leq n$, $\Sigma_{\src(e) \cap V(\theta_{r}) \neq \emptyset} x_e = 0$ --- this condition ensures that at least one set $V(\theta_{r})$ does not have any state in the cycle associated with the solution. 
	
\end{enumerate}

In this case, $ \LP'(\theta_{1}, \ldots, \theta_{n}) $ has a solution if and only if there is a path $ \pi $ such that $ z_i \leq \pay_i(\pi) $ for every player $ i $ and, for at least one $ V(\theta_{r}) $, its states are visited only \textit{finitely many times}. Thus, we have a procedure that checks if there is a path $ \pi $ that satisfies $ \lnot \phi $ such that $ z_i \leq \pay_i(\pi) $ for every player $ i $, if and only if both linear programs have a solution. Using this new construction, we can now prove the following result. 

\begin{theorem}\label{thm:strong-gr1}
	 \strong with \GRone specifications  is $ \SigmaPTwo $-complete.
\end{theorem}

\begin{proof}[Proof sketch]
	For membership, observe that by rearranging the problem statement, we have the following question:
	Check whether the following expression is true
	\begin{ceqn}
		\begin{align*}
		\exists \kappa \in \K(\Game,\beta)&,\tag{1}\\
		&\exists \strpElm \in \strElm_1 \times \cdots \times \strElm_n,~\text{such that}~\strpElm \in \NE(\Game,\kappa),\tag{2}\\
		~\text{and}&\\
		&\forall \strpElm' \in \strElm_1 \times \cdots \times \strElm_n,~\text{if}~\strpElm' \in \NE(\Game,\kappa)~\text{then}~\pi(\strpElm') \models \phi.\tag{3}
		\end{align*}
	\end{ceqn}
	Statement $ (2) $ can be checked in \np (Theorem \ref{thm:pthfinding}), whereas verifying statement $ (3) $ is in co\np; to see this, notice that we can rephrase $ (3) $ as follows: $ \neg \exists z \in \{ \pun_i(s) : s \in \St \}^{\Ag}  $ such that both $ \LP(\psi_{l}) $ and $ \LP'(\theta_1,\dots,\theta_{n}) $ have a solution in $ (\Game,\kappa){[z]} $. Thus, membership in $ \SigmaPTwo $ follows.
	We prove hardness via a reduction from $ \QSAT_2 $ (satisfiability of quantified Boolean formulae with 2 alternations), which is known to be $ \SigmaPTwo $-complete  \cite{1994-papadimitriou}. 
\end{proof}

\section{Optimality and Uniqueness of Solutions}
Having asked the questions studied in the previous sections, the principal -- the {\em designer} in the equilibrium design problem -- may want to explore further information. Because the power of the principal is limited by its budget, and because from the point of view of the system, it may be associated with a reward ({\em e.g.}, money, savings, etc.) or with the inverse of the amount of a finite resource ({\em e.g.}, time, energy, etc.) an obvious question is asking about {\em optimal} solutions. This leads us to {\em optimisation} variations of the problems we have studied. Informally, in this case, we ask what is the least budget that the principal needs to ensure that the implementation problems have positive solutions. The principal may also want to know whether a given subsidy scheme is {\em unique}, so that there is no point in looking for any other solutions to the problem. In this section, we investigate these kind of problems, and classify our study into two parts, one corresponding to the \weak problem and another one corresponding to the \strong problem.

\subsection{Optimality and Uniqueness in the Weak Domain}
We can now define formally some of the problems that we will study in the rest of this section. To start, the optimisation variant for \weak is defined as follows.

\begin{definition}[\optwi]
	Given a game $\Game$ and a specification formula $\phi$:
	\begin{center}
		What is the optimum budget $ \beta$ such that $ \WI(\Game,\phi,\beta) \neq \varnothing $?
	\end{center}
\end{definition}

Another natural problem, which is related to \optwi, is the ``exact'' variant -- a membership question. In this case, in addition to $ \Game $ and $ \phi $, we are also given an integer $ b $, and ask whether it is indeed the smallest amount of budget that the principal has to spend for some optimal weak implementation. This decision problem is formally defined as follows.

\begin{definition}[\exactwi]
	Given a game $\Game$, a specification formula $\phi$, and an integer $ b $:
	\begin{center}
		Is $b$ equal to the optimum budget for $ \WI(\Game,\phi,\beta) \neq \varnothing $?
	\end{center}
\end{definition}

To study these problems, it is useful to introduce some concepts first. More specifically, let us introduce the concept of {\em implementation efficiency}. We say that a \weak (resp.\ \strong) is \textit{efficient} if $ \beta = \cost(\kappa) $ and there is no $ \kappa' $ such that $ \cost(\kappa') < \cost(\kappa) $ and $ \kappa' \in \WI(\Game,\phi,\beta) $ (resp.\ $ \kappa' \in \SI(\Game,\phi,\beta) $). In addition to the concept of efficiency for an implementation problem, it is also useful to have the following result.


\begin{proposition}\label{lem:opt-bound}
Let $ z_i $ be the largest payoff that player $ i $ can get after deviating from a path $\pi$. The optimum budget is an integer between 0 and $ \sum_{i \in \Ag} z_i \cdot (|\St|-1) $.
\end{proposition}


Using Proposition~\ref{lem:opt-bound}, we can show that both \optwi and \exactwi can be solved in PSPACE for \LTL specifications. Intuitively, the reason is that we can use the upper bound given by Proposition~\ref{lem:opt-bound} to go through all possible solutions in exponential time, but using only nondeterministic polynomial space. Formally, we have the following results. 

\begin{theorem}\label{thm:optwi-ltl}
	\optwi with \LTL specifications is \fpspace-complete.
\end{theorem}


\begin{corollary}\label{cor:exactwi-ltl}
	\exactwi with \LTL specifications is \pspace-complete.
\end{corollary}

The fact that both \optwi and \exactwi with \LTL specifications can be answered in, respectively, \fpspace and \pspace does not come as a big surprise: checking an instance can be done using polynomial space and there are only exponentially many instances to be checked. However, for \optwi and \exactwi with \GRone specifications, these two problems are more interesting. 

\begin{theorem}\label{thm:optwi-gr1}
	\optwi with \GRone specifications is $ \FP^{\np} $-complete.
\end{theorem}

\begin{proof}[Proof sketch]
	Membership follows from the fact that the search space, which is bounded as in Proposition \ref{lem:opt-bound}, can be explored using binary search and \weak as an oracle. More precisely, we can find the smallest budget $ \beta $ such that $ \WI(\Game,\phi,\beta) \neq \varnothing $ by checking every possible value for $ \beta $, which lies between 0 and $ 2^n $, where $ n $ is the length of the encoding of the instance. Since, due to the binary search routine, we need logarithmically many calls to the \np oracle ({\em i.e.}, to \weak), in the end we have a searching procedure that would run in polynomial time. 
	For the lower bound, we reduce from \textsc{TSP Cost} (the optimal travelling salesman problem), which is $ \FP^{\np} $-complete \cite{1994-papadimitriou}. 
	\end{proof}

\begin{corollary}\label{cor:exactwi-gr1}
	\exactwi with \GRone specifications is \DP-complete.
\end{corollary}

\begin{proof}
	For membership, observe that an input is a ``yes'' instance of \exactwi if and only if it is a ``yes'' instance of \weak~{\em and} a ``yes'' instance of \weakc (the problem where one asks whether $ \WI(\Game,\phi,\beta) = \varnothing $). Since the former problem is in \np and the latter problem is in co\np, membership in \DP follows. For the lower bound, we use the same reduction technique as in Theorem~\ref{thm:optwi-gr1}, and reduce from \textsc{Exact TSP}, a problem known to be \DP-hard \cite{1994-papadimitriou,PAPADIMITRIOU1984244}.
\end{proof}

Following \cite{Papadimitriou1984}, we may naturally ask whether the optimal solution given by \optwi is unique. We call this problem \uoptwi. For some fixed budget $ \beta $, it may be the case that for two subsidy schemes $ \kappa, \kappa' \in \WI(\Game,\phi,\beta) $ -- we assume the implementation is efficient -- we have $ \kappa \neq \kappa'$ and $ \cost(\kappa) = \cost(\kappa') $. With \LTL specifications, it is not difficult to see that we can solve \uoptwi in polynomial space. Therefore, we have the following result.

\begin{corollary}\label{cor:uoptwi-ltl}
	\uoptwi with \LTL specifications is \pspace-complete.
\end{corollary}

For \GRone specifications, we reason about \uoptwi using the following procedure:
	\begin{enumerate}
	\item Find the exact budget using binary search and \weak as an oracle;
	\item Use an \np oracle once to guess two distinct subsidy schemes with precisely this budget; if no such subsidy schemes exist, return ``yes''; otherwise, return ``no''.
\end{enumerate}
The above decision procedure clearly is in $ \DeltaPTwo $ (for the upper bound). Furthermore, since Theorem \ref{thm:optwi-gr1} implies $ \DeltaPTwo $-hardness \cite{KRENTEL1988490} (for the lower bound), we have the following corollary.

\begin{corollary}\label{cor:uoptwi-gr1}
	\uoptwi with \GRone specifications is $ \DeltaPTwo $-complete.
\end{corollary}

\subsection{Optimality and Uniqueness in the Strong Domain}
In this subsection, we study the same problems as in the previous subsection but with respect to the \strong variant of the equilibrium design problem. We first formally define the problems of interest and then present the two first results.

\begin{definition}[\optsi]
	Given a game $\Game$ and a specification formula $\phi$:
	\begin{center}
		What is the optimum budget $ \beta$ such that $ \SI(\Game,\phi,\beta) \neq \varnothing $?
	\end{center}
\end{definition}

\begin{definition}[\exactsi]
	Given a game $\Game$, a specification formula $\phi$, and an integer $ b $:
	\begin{center}
		Is $b$ equal to the optimum budget for $ \SI(\Game,\phi,\beta) \neq \varnothing $?
	\end{center}
\end{definition}

For the same reasons discussed in the weak versions of these two problems, we can prove the following two results with respect to games with \LTL specifications. 

\begin{theorem}\label{thm:optsi-ltl}
	\optsi with \LTL specifications is \fpspace-complete.
\end{theorem}


\begin{corollary}\label{cor:exactsi-ltl}
	\exactsi with \LTL specifications is \pspace-complete.
\end{corollary}

For \GRone specifications, observe that using the same arguments for the upper-bound of \optwi with \GRone specifications, we obtain the upper-bound for \optsi with \GRone specifications. Then, it follows that \optsi is in $ \FP^{\SigmaPTwo} $. For hardness, we define an $ \FP^{\SigmaPTwo} $-complete problem, namely $ \MWQSAT_2 $. Recall that in $ \QSAT_2 $ we are given a Boolean 3DNF formula $ \psi(\setx,\sety)$ and sets $\setx =\{ x_1,\dots,x_n \}, \sety = \{ y_1,\dots,y_m \} $, with a set of terms $T =\{ t_1,\dots,t_k \} $. Define $ \MWQSAT_2 $ as follows. Given $ \psi(\setx,\sety) $ and a weight function $ \cFun: \setx \to \mathbb{Z}^{\geq}$, $ \MWQSAT_2 $ is the problem of finding an assignment $ \vecx \in \{0,1\}^n $ with the least total weight such that $ \psi(\setx,\sety) $ is true for every $ \vecy \in \{0,1\}^m $. Observe that $ \MWQSAT_2 $ generalises $ \MQSAT_2 $, which is known to be $ \FP^{\SigmaPTwo[\log n]} $-hard \cite{ChocklerH04}, {\em i.e.}, $ \MQSAT_2 $ is an instance of $ \MWQSAT_2 $, where all weights are~1.


\begin{theorem}
	$ \MWQSAT_2 $ is $ \FP^{\SigmaPTwo} $-complete.
\end{theorem}

\begin{proof}
	Membership follows from the upper-bound of $ \MQSAT_2 $ \cite{ChocklerH04}: since we have an exponentially large input with respect to that of $ \MQSAT_2 $, by using binary search we will need polynomially many calls to the $ \SigmaPTwo $ oracle. Hardness is immediate~\cite{ChocklerH04}.
\end{proof}

Now that we have an $ \FP^{\SigmaPTwo} $-hard problem in our hands, we can proceed to determine the complexity class of \optsi with \GRone specifications. 
For the upper bound we one can use arguments analogous to those in Theorem \ref{thm:optwi-gr1}. For the lower bound, one can reduce from $ \MWQSAT_2 $. Formally, we have: 

\begin{theorem}\label{thm:optsi-gr1}
	\optsi with \GRone specifications is $ \FP^{\SigmaPTwo} $-complete.
\end{theorem}

\begin{corollary}\label{cor:exactsi-gr1}
	\exactsi with \GRone specifications is $ \DPTwo $-complete.
\end{corollary}
\begin{proof}
	Membership follows from the fact that an input is a ``yes'' instance of \exactsi (with \GRone specifications) if and only if it is a ``yes'' instance of \strong~{\em and} a ``yes'' instance of \strongc, the decision problem where we ask $ \SI(\Game,\phi,\beta) = \varnothing $ instead. The lower bound follows from the hardness of \strong and \strongc problems, which immediately implies \DPTwo-hardness \cite[Lemma 3.2]{Aleksandrowicz2017}.
\end{proof}

Furthermore, analogous to \uoptwi, we also have the following corollaries. 

\begin{corollary}\label{cor:uoptsi-ltl}
	\uoptsi with \LTL specifications is \pspace-complete.
\end{corollary}

\begin{corollary}\label{cor:uoptsi-gr1}
	\uoptsi with \GRone specifications is $ \DeltaPThree $-complete.
\end{corollary}

\section{Conclusions \& Related and Future Work}

\paragraph*{Equilibrium design vs.\ mechanism design -- connections with Economic theory.}
Although equilibrium design is closely related to mechanism design, as typically studied in game theory~\cite{HR06}, the two are not exactly the same. Two key features in mechanism design are the following. Firstly, in a mechanism design problem, the designer is not given a game structure, but instead is asked to provide one; in that sense, a mechanism design problem is closer to a rational synthesis problem~\cite{FismanKL10,GutierrezHW15}. Secondly, in a mechanism design problem, the designer is only interested in the game's outcome, which is given by the payoffs of the players in the game; however, in equilibrium design, while the designer is interested in the payoffs of the players as these may need to be perturbed by its budget, the designer is also interested -- and in fact primarily interested -- in the satisfaction of a temporal logic goal specification, which the players in the game do not take into consideration when choosing their individual rational choices; in that sense, equilibrium design is closer to rational verification~\cite{GutierrezHW17-aij} than to mechanism design. Thus, equilibrium design is a new computational problem that sits somewhere in the middle between mechanism design and rational verification/synthesis. Technically, in equilibrium design we go beyond rational synthesis and verification through the additional design of subsidy schemes for incentivising behaviours in a concurrent and multi-agent system, but we do not require such subsidy schemes to be incentive compatible mechanisms, as in mechanism design theory, since the principal may want to reward only a group of players in the game so that its temporal logic goal is satisfied, while rewarding other players in the game in an unfair way -- thus, leading to a game with a suboptimal social welfare measure. In this sense, equilibrium design falls short with respect to the more demanding social welfare requirements often found in mechanism design theory.  

\paragraph*{Equilibrium design vs.\ rational verification -- connections with Computer science.}
Typically, in rational synthesis and verification~\cite{FismanKL10,GutierrezHW15,GutierrezHW17-aij,KupfermanPV16} we want to check whether a property is satisfied on some/every Nash equilibrium computation run of a reactive, concurrent, and multi-agent system. These verification problems are primarily concerned with qualitative properties of a system, while assuming rationality of system components. However, little attention is paid to quantitative properties of the system. This drawback has been recently identified and some work has been done to cope with questions where both qualitative and quantitative concerns are considered~\cite{AKP18,BBFR13,ChatterjeeD12,CDHR10,ChatterjeeHJ05,GMPRW17,GNPW19,VCDHRR15}. Equilibrium design is new and different approach where this is also the case. More specifically, as in a mechanism design problem, through the introduction of an external principal -- the designer in the equilibrium design problem -- we can account for overall qualitative properties of a system (the principal's goal given by an \LTL or a \GRone specification) as well as for quantitative concerns (optimality of solutions constrained by the budget to allocate additional rewards/resources). Our framework also mixes qualitative and quantitative features in a different way: while system components are only interested in maximising a quantitative payoff, the designer is primarily concerned about the satisfaction of a qualitative (logic) property of the system, and only secondarily about doing it in a quantitatively optimal way. 

\paragraph*{Equilibrium design vs.\ repair games and normative systems -- connections with AI.}
In recent years, there has been an interest in the analysis of rational outcomes of multi-agent systems modelled as multi-player games. This has been done both with modelling and with verification purposes. In those multi-agent settings, where AI agents can be represented as players in a multi-player game, a focus of interest is on the analysis of (Nash) equilibria in such games~\cite{BouyerBMU15,GutierrezHW17-aij}. However, it is often the case that the existence of Nash equilibria in a multi-player game with temporal logic goals may not be guaranteed~\cite{GutierrezHW15,GutierrezHW17-aij}. For this reason, there has been already some work on the introduction of desirable Nash equilibria in multi-player games~\cite{AlmagorAK15,Perelli19}. This problem has been studied as a repair problem~\cite{AlmagorAK15} in which either the preferences of the players (given by winning conditions) or the actions available in the game are modified; the latter one also being achieved with the use of normative systems~\cite{Perelli19}. In equilibrium design, we do not directly modify the preferences of agents in the system, since we do not alter their goals or choices in the game, but we indirectly influence their rational behaviour by incentivising players to visit, or to avoid, certain states of the overall system. We studied how to do this in an (individually) optimal way with respect to the preferences of the principal in the equilibrium design problem. However, this may not always be possible, for instance, because the principal's temporal logic specification goal is just not achievable, or because of constraints given by its limited budget. 

\paragraph*{Future work: social welfare requirements and practical implementation.}
As discussed before, a key difference with mechanism design is that social welfare requirements are not considered~\cite{MSZ13}. However, a benevolent principal might not see optimality as an individual concern, and instead consider the welfare of the players in the design of a subsidy scheme. In that case, concepts such as the  \textit{utilitarian social welfare} may be undesirable as the social welfare maximising the payoff received by players might allocate all the budget to only one player, and none to the others. A potentially better option is to improve fairness in the allocation of the budget by maximising the \textit{egalitarian social welfare}. Finally, given that the complexity of equilibrium design is much better than that of rational synthesis/verification, we should be able to have efficient implementations, for instance, as an extension of EVE~\cite{GutierrezNPW18}. 

\bibliography{incentive}

\appendix
\section{Proofs}

\paragraph*{Proof of Theorem \ref{thm:pthfinding}}
\begin{proof}
	For (1) implies (2): Let $ z_i $ be the largest value player $ i $ can get by deviating from $ \pi $. Let $ k \in \SetN $ be such that $ z_i = \pun_i(\trnFun(s_k,(\jact_{-i},\act'_i))) $. Suppose further that $ \pay_i(\pi) < z_i $. Thus, player $ i $ would deviate at $ s_k $, which is a contradiction to $ \pi $ being a path induced by a Nash equilibrium.
	
	For (2) imples (1): Define strategy profile $ \strpElm $ that follows $ \pi $ as long as no-one has deviated from $ \pi $.
	In such a case where player $ i $ deviates on the $k$-th iteration, the strategy profile $\strpElm_{-i}$ starts playing the $z_i$-secure strategy for player $i$ that guarantees the payoff of player $i$ to be less than $z_i$. Therefore, we have $\pay_i(\pi(\strpElm_{-i},\strElm'_{i})) \leq z_i \leq \pay_i(\pi)$, for every possible strategy $\strElm'_{i}$ of player $i$ (the second inequality is due to condition $2(b)$). Thus, there is no beneficial deviation for player $ i $ and $ \pi $ is a path induced by a Nash equilibrium.
\end{proof}

\paragraph*{Proof of Proposition \ref{prop:kbound}}
\begin{proof}
	
	For a fixed budget $b$, the number of subsidy schemes of budget exactly $b$ corresponds to the number of \emph{weak compositions} of $b$ in $m$ parts, which is given by $\binom{b+m-1}{b}$~\cite{HT09}.
	Therefore, the number of subsidy schemes of budget at most $\beta$ is the sum
	
	$$
	\card{\K(\Game,\beta)} = \sum_{b=0}^{\beta}\binom{b+m-1}{b}\text{.}
	$$
	
	We now prove that 
	
	$$\sum_{b=0}^{\beta}\binom{b+m-1}{b} = \frac{\beta + 1}{m} \binom{\beta + m}{\beta + 1}\text{.}$$
	
	By induction on $\beta$, as base case, for $\beta = 0$, we have that 
	
	$$
	\binom{\beta + m - 1}{\beta} = 1 = \frac{\beta + 1}{m} \binom{\beta +m}{\beta +1}\text{.}
	$$
	
	For the inductive case, let us assume that the assertion hold for some $\beta$ and let us prove for $\beta + 1$.
	We have the following:
	
	$$\sum_{b=0}^{\beta + 1}\binom{b+m-1}{b} = \sum_{b=0}^{\beta}\binom{b+m-1}{b} + \binom{\beta + m - \cancel{1} + \cancel{1}}{\beta + 1} = \frac{\beta + 1}{m} \binom{\beta + m}{\beta + 1} + \binom{\beta + m}{\beta + 1} \text{.}$$
	
	Therefore we have
	
	\begin{align*}
		\frac{\beta + 1}{m} \binom{\beta + m}{\beta + 1} + \binom{\beta + m}{\beta + 1} & = \binom{\beta + m}{\beta + 1} \left(\frac{\beta + 1}{m} + 1\right) = \\
		\binom{\beta + m}{\beta + 1} \frac{\beta + 1 + m}{m} 
		& = \frac{\beta + 1 + m}{m} \cdot \frac{(\beta + m)!}{(\beta + 1)!(\cancel{\beta} + m - \cancel{\beta} - 1)!} = \\
		\frac{(\beta + m + 1)!}{(\beta + 1)!m!} = \frac{(\beta + m + 1)!}{(\beta + 1)!m!}\cdot \frac{\beta + 2}{\beta + 2} \cdot \frac{m}{m} & = \frac{\beta + 2}{m} \cdot \frac{(\beta + m + 1)!}{(\beta + 2)! (m - 1)!} = \\
		\frac{\beta + 2}{m} \cdot \frac{(\beta + m + 1)!}{(\beta + 2)! (\beta + m + 1 - \beta - 2)!} & = \frac{\beta + 2}{m} \binom{\beta + m + 1}{\beta + 2}
	\end{align*}
	
	that proves the assertion.
\end{proof}

\paragraph*{Proof of Theorem~\ref{thm:weak-gr1}}
\begin{proof}
	We will define a linear program of size polynomial in $(\Game,\kappa)$ having a solution if and only if there exists an ultimately periodic path $ \pi $ such that $z_i \leq \pay_i(\pi) $ and satisfies the \GRone specification. 
	
	Recall that $\varphi$ has the following form
	$$
	\varphi = \bigwedge_{l = 1}^{m} \always \sometime \psi_{l} \to \bigwedge_{r = 1}^{n} \always \sometime \theta_{r}\text{,}
	$$
	and let $V(\psi_{l})$ and $V(\theta_r)$ be the subset of states in $(\Game,\kappa)$ that satisfy the boolean combinations $\psi_{l}$ and $\theta_{r}$, respectively. Observe that property $\varphi$ is satisfied over a path $\pi$ if, and only if, either $\pi$ visits every $V(\theta_r)$ infinitely many times or visits some of the $V(\psi_{l})$ only a finite number of times.
	
	For the game $(\Game,\kappa){[z]}$, let $ W = (V, E, (\wFun_a)_{a \in \Ag})$ be the underlying multi-weighted graph, and for every edge $e\in E$ introduce a variable $x_e$. Informally, the value $x_e$ is the number of times that the edge $e$ is used on a cycle. Formally, let $\src(e) = \{v \in V : \exists w\, e = (v,w) \in E\}$; $\trg(e) = \{v \in V : \exists w\, e = (w,v) \in E\}$; $\OUT(v) = \{e \in E : \src(e) = v\}$; and $\IN(v) = \{e \in E : \trg(e) = v\}$.
	
	Consider $\psi_{l}$ for some $1 \leq l \leq m$, and define the linear program $\LP(\psi_{l})$ with the following inequalities and equations:
	\begin{enumerate}
		\item[Eq1:]
		$x_e \geq 0$ for each edge $e$ 
		--- a basic consistency criterion;
		
		\item[Eq2:]
		$\Sigma_{e \in E} x_e \geq 1$ 
		--- ensures that at least one edge is chosen;
		
		\item[Eq3:]
		for each $a \in \Ag$, $\Sigma_{e \in E} \wFun_a(\src(e)) x_e \geq 0$ --- this enforces that the total sum of any solution is non-negative;
		
		\item[Eq4:]
		$\Sigma_{\src(e) \cap V(\psi_{l}) \neq \emptyset} x_e = 0$ --- this ensures that no state in $V(\psi_{l})$ is in the cycle associated with the solution;
		
		\item[Eq5:]
		for each $v \in V$, $\Sigma_{e \in \OUT(v)} x_e = \Sigma_{e \in \IN(v)} x_e$  --- this condition says that the number of times one enters a vertex is equal to the number of times one leaves that vertex.	
	\end{enumerate}
	
	By construction, it follows that $\LP(\psi_{l})$ admits a solution if and only if there exists a path $\pi$ in $\Game$ such that $z_i \leq \pay_i(\pi) $ for every player $i$ and visits $V(\psi_{l})$ only {\em finitely many times}.
	In addition, consider the linear program $\LP(\theta_{1}, \ldots, \theta_{n})$ defined with the following inequalities and equations: 
	
	\begin{enumerate}
		\item[Eq1:]
		$x_e \geq 0$ for each edge $e$ --- 
		a basic consistency criterion;
		
		\item[Eq2:]
		$\Sigma_{e \in E} x_e \geq 1$ --- 
		ensures that at least one edge is chosen;
		
		\item[Eq3:]
		for each $a \in \Ag$, $\Sigma_{e \in E} \wFun_a(\src(e)) x_e \geq 0$ --- this enforces that the total sum of any solution is non-negative;
		
		\item[Eq4:]
		for all $1 \leq r \leq n$, $\Sigma_{\src(e) \cap V(\theta_{r}) \neq \emptyset} x_e \geq 1$ --- this ensures that for every $V(\theta_{r})$ at least one state is in the cycle;
		
		\item[Eq5:]
		for each $v \in V$, $\Sigma_{e \in \OUT(v)} x_e = \Sigma_{e \in \IN(v)} x_e$  --- this 
		condition says that the number of times one enters a vertex is equal to the number of times one leaves that vertex.	
	\end{enumerate}
	
	In this case, $\LP(\theta_{1}, \ldots, \theta_{n})$  admits a solution if and only if there exists a path $\pi$ such that $z_i \leq \pay_i(\pi) $ for every player $i$ and visits every $V(\theta_{r})$ {\em infinitely many times}.
	
	Since the constructions above are polynomial in the size of both $(\Game,\kappa)$ and $\phi$, we can conclude it is possible to check in \np the statement that there is a path $\pi$ satisfying $\varphi$ such that $z_i \leq \pay_i(\pi) $ for every player~$i$ in the game if and only if one of the two linear programs defined above has a solution.
		
	For the lower-bound, we use \cite{UW11} and observe that if $\phi$ is true and $ \beta = 0 $, then the problem is equivalent to checking whether the \MP game has a Nash equilibrium. 
\end{proof}

\paragraph*{Proof of Theorem \ref{thm:strong-gr1}}
\begin{proof}
	For membership, observe that by rearranging the problem statement, we have the following question:\\
	Check whether the following expression is true
	\begin{ceqn}
		\begin{align*}
		\exists \kappa \in \K(\Game,\beta)&,\tag{1}\\
		&\exists \strpElm \in \strElm_1 \times \cdots \times \strElm_n,~\text{such that}~\strpElm \in \NE(\Game,\kappa),\tag{2}\\
		~\text{and}&\\
		&\forall \strpElm' \in \strElm_1 \times \cdots \times \strElm_n,~\text{if}~\strpElm' \in \NE(\Game,\kappa)~\text{then}~\pi(\strpElm') \models \phi.\tag{3}
		\end{align*}
	\end{ceqn}
	Statement $ (2) $ can be checked in \np (Theorem \ref{thm:pthfinding}). Whereas, verifying statement $ (3) $ is in co\np; to see this, notice that we can rephrase $ (3) $ as follows: $ \not \exists z \in \{ \pun_i(s) : s \in \St \}^{\Ag}  $ such that both $ \LP(\psi_{l}) $ and $ \LP'(\theta_1,\dots,\theta_{n}) $ have a solution in $ (\Game,\kappa){[z]} $. Thus $ \SigmaPTwo $ membership follows.

	We prove hardness by a reduction from $ \QSAT_2 $ (satisfiability of quantified Boolean formula with 2 alternations) \cite{1994-papadimitriou}. Let $ \psi(\setx,\sety) $ be an $ n + m $ variable Boolean 3DNF formula, where $ \setx = \{x_1,\dots,x_n \} $ and $ \sety = \{ y_1,\dots, y_n \} $, with $ t_1,\dots,t_k $ terms. Write $ \term_j $ for the set of literals in $ j $-th term and $ \term_j^i $ for the $ i $-th literal in $ \term_j $. Moreover write $ x^j_i $ and $ y^j_i $ for variable $ x_i \in \setx $ and $ y_i \in \sety $ that appears in $ j $-th term, respectively. For instance, if the fifth term is of the form of $ (x_2 \wedge \lnot x_3 \wedge y_4) $, then we have $ \term_5 = \{ x^5_2, x^5_3, y^5_4 \} $ and $ \term_5^1 = x^5_2 $. Let $ \T = \{ \term_i \cap \sety : 1 \leq i \leq k \} $, that is, the set of subset of $ \term_i $ that contains only $ y $-literals.
	
	For a formula $ \psi(\setx,\sety) $ we construct an instance of \strong such that $ \SI(\Game,\phi,\beta) \neq \varnothing $ if and only if there is an $ \vecx \in \{0,1\}^n $ such that $ \psi(\setx,\sety) $ is true for every $ \vecy \in \{0,1\}^m $. Let $ \Game $ be such a game where
	
	\begin{itemize}
		\item $ \Ag = \{1,2\} $,
		\item $ \St = \{ \bigcup_{j \in [1,k]} (\term_j \times \{0,1\}^3) \} \cup \{\T \times \{0\}^3 \} \cup \{ \tuple{\source,\{0\}^3}, \tuple{\sink,\{0\}^3} \} $,
		\item $ s_0 = \source $,
		\item for each state $ s \in \St $
		\begin{itemize}
			\item $ \Ac_1(s) = \{ \{ \T \cup \{\sink \} \} \times \{0\}^3 \} $, $ \Ac_2(s) = \{ \varepsilon \} $, if $ s = \tuple{\source,\{0\}^3} $,
			\item $ \Ac_1(s) = \{ \term^1_i : s[0] \subseteq \term_i \wedge i \in [1,k] \} $, $ \Ac_2(s) = \{ 0,1 \}^{3}$, if $ s \in \{\T \times \{0\}^3 \} $,
			\item $ \Ac_1(s) = \{ \varepsilon \} $, $ \Ac_2(s) = \{ \varepsilon \} $, if $ s \in \bigcup_{j \in [1,k]} (\term_j \times \{0,1\}^3) $,
		\end{itemize}
		\item for an action profile $ \jact = (\act_{1},\act_{2}) $
		\begin{itemize}
			\item $ \trnFun(s,\jact) = \act_{1} $, if $ s = \tuple{\source,\{0\}^3} $,
			\item $ \trnFun(s,\jact) = \tuple{\act_{1},\act_{2}} $, if $ s \in \{\T \times \{0\}^3 \} $,
			\item $ \trnFun(s,\jact) = \tuple{\term_j^{(i \bmod 3)+1 },s[1]} $, if $ s = \tuple{\term^i_j,s[1]} \in \bigcup_{j \in [1,k]} (\term_j \times \{0,1\}^3) $;
			\item $ \trnFun(s,\jact) = s $, otherwise;
		\end{itemize}
		\item for each state $ s \in \St, \labFun(s) = s[0] $,
		\item for each state $ s \in \St $
		\begin{itemize}
			\item $ \wFun_1(s) = \frac{2}{3} $, if $ s[0] = \sink $\footnote{This can be implemented by a macrostate with three substates---2 substates with weight of 1, and 1 with weight of 0---forming a simple cycle.},
			\item $ \wFun_1(s) = 0 $, otherwise;
		\end{itemize}
		\item the payoff of player $ i \in \Ag $ for an ultimately periodic path $ \pi $ in $ \Game $ is
		\begin{itemize}
			\item $ \pay_1(\pi) = \MP(\wFun_1(\pi)) $,
			\item $ \pay_2(\pi) = - \MP(\wFun_1(\pi)) $,
		\end{itemize}
	\end{itemize}
	Furthermore, let $ \beta = |\setx| $ and the \GRone property to be $ \phi := \always \sometime~\lnot \sink $. Define a (partial) subsidy scheme $ \kappa : \setx \to \{0,1\} $. The weights are updated with respect to $ \kappa $ as follows:\\	
	for each $ s \in \St $ such that $ s[0] \in \term_j \setminus \sety $, that is, an $ x $-literal that appears in term $ t_j $
	\[
	\wFun_1(s) =
	\begin{cases}
	1, & \text{if } \kappa(s) = 1 \wedge s[0]~\text{is not negated in}~t_j \\
	1, & \text{if } \kappa(s) = 0 \wedge s[0]~\text{is negated in}~t_j \\
	0, & \text{if } \kappa(s) = 1 \wedge s[0]~\text{is negated in}~t_j \\
	0, & \text{otherwise;}
	\end{cases}
	\]
	for each $ s \in \St $ such that $ s[0] \in \term_j \setminus \setx $, that is, a $ y $-literal that appears in term $ t_j $, $ s[0] = \term^i_j $
	\[
	\wFun_1(s) =
	\begin{cases}
	1, & \text{if } s[1][i] = 1 \wedge s[0]~\text{is not negated in}~t_j \\
	1, & \text{if } s[1][i] = 0 \wedge s[0]~\text{is negated in}~t_j \\
	0, & \text{if } s[1][i] = 1 \wedge s[0]~\text{is negated in}~t_j \\
	0, & \text{otherwise;}
	\end{cases}
	\]
	the weights of other states remain unchanged.
	
	The construction is now complete, and polynomial to the size of formula $ \psi(\setx,\sety) $. We claim that $ \SI(\Game,\phi,\beta) \neq \varnothing $ if and only if there is an $ \vecx \in \{0,1\}^n $ such that $ \psi(\setx,\sety) $ is true for every $ \vecy \in \{0,1\}^m $. From left to right, consider a subsidy scheme $ \kappa \in \SI(\Game,\phi,\beta) $ which implies that there exists no Nash equilibrium run in $ (\Game,\kappa) $ that ends up in $ \sink $. This means that for every action $ \jact_{2} \in \Ac_2(s), $ there exists $ \jact_1 \in \Ac_1(s), s \in \{\T \times \{0\}^3 \} $, such that $ \pay_1(\pi) = 1 $, where $ \pi $ is the resulting path of the joint action. Observe that this corresponds to the existence of (at least) a term $ t_i $, which evaluates to true under assignment $ \vecx $, regardless the value of $ \vecy $. From right to left, consider an assigment $ \vecx \in \{0,1\}^n $ such that for all $ \vecy \in \{0,1\}^m $, the formula $ \psi(\setx,\sety) $ is true. This means that for every $ \vecy $, there exists (at least one) term $ t_i $ in $ \psi(\setx,\sety) $ that evaluates to true. By construction, specifically the weight updating rules, for every $ \jact_{2} $ corresponding to assignment $ \vecy $, there exists $ \term_j $ such that $ \forall i \in [1,3], \wFun_1(\term^i_j) = 1 $. This means that player 1 can always get payoff equals to 1, therefore, any run that ends in $ \sink $ is not sustained by Nash equilibrium.
\end{proof}

\paragraph*{Proof of Proposition~\ref{lem:opt-bound}}
\begin{proof}
	The lower-bound is straightforward. The upper-bound follows from the fact that the maximum value the principal has to pay to player $ i $ is when the path $ \pi $ is a simple cycle and formed from all states in $ \St $, apart from 1 deviation state. 
\end{proof}

\paragraph*{Proof of Theorem \ref{thm:optwi-ltl}}
\begin{proof}
	Since the search space is bounded (Proposition \ref{lem:opt-bound}), by using \weak an an oracle we can iterate through every instance and return the smallest $ \beta $ such that $ \WI(\Game,\phi,\beta) \neq \varnothing $. Moreover, each instance is of polynomial size in the size of the input. Thus membership in \pspace follows. Hardness is straightforward.
\end{proof}

\paragraph*{Proof of Theorem \ref{thm:optwi-gr1}}
\begin{proof}
	Membership follows from the fact that the search space, which is bounded as in Proposition \ref{lem:opt-bound}, can be fully explored using binary search and \weak as an oracle. More precisely, we can find the smallest budget $ \beta $ such that $ \WI(\Game,\phi,\beta) \neq \varnothing $ by checking every possible value for $ \beta $, which lies between 0 and $ 2^n $, where $ n $ is the length of the encoding of the instance. Since we need logarithmically many calls to the \np oracle (to \weak), in the end we have searching procedure that runs in polynomial time. 
	
	For hardness we reduce from \textsc{TSP Cost} (optimal travelling salesman problem) that is known to be $ \FP^{\np} $-complete \cite{1994-papadimitriou}. Given a \textsc{TSP Cost} instance $ \tuple{G, c} $, $ G = \tuple{V, E} $ is a graph, $ c : E \to \SetZ $ is a cost function. We assume that $ \WI(\Game,\phi,\beta) $ is efficient. To encode \textsc{TSP Cost} instance, we construct a game $ \Game $ and \GRone formula $ \phi $, such that the optimum budget $ \beta $ corresponds to the value of optimum tour. Let $ \Game $ be such a game where
	\begin{itemize}
		\item $ \Ag = \{ 1 \} $,
		\item $ \St = \{ \tuple{v,e} : v \in V \wedge e \in \IN(v) \} \cup \{ \tuple{\sink,\varepsilon} \} $,
		\item $ s_0 $ can be chosen arbitrarily from $ \St \setminus \{\tuple{\sink,\varepsilon}\} $,
		\item for each state $ \tuple{v,e} \in \St $ and edge $ e' \in E \cup \{\varepsilon\} $
		\begin{itemize}
			\item $ \trnFun(\tuple{v,e},e') = \tuple{\trg(e'),e'}, \text{if}~v \neq \sink~\text{and}~e' \neq \varepsilon, $
			\item $ \trnFun(\tuple{v,e},e') = \tuple{\sink,\varepsilon}, \text{otherwise}; $
		\end{itemize}			
		\item for each state $ \tuple{v,e} \in \St $
		\begin{itemize}
			\item $ \wFun_1(\tuple{v,e}) = \max\{c(e'): e' \in E \} - c(e) $, if $ v \neq \sink $,
			\item $ \wFun_1(\tuple{v,e}) = \max\{c(e'): e' \in E \} $, otherwise;
		\end{itemize}
		\item the payoff of player 1 for a path $ \pi $ in $ \Game $ is $ \pay_1(\pi) = \MP[\wFun_1(\pi)] $,
		\item for each state $ \tuple{v,e} \in \St $, the set of actions available to player 1 is $ \OUT(v) \cup \{\varepsilon\} $,
		\item for each state $ \tuple{v,e} \in \St $, $ \labFun(\tuple{v,e}) = v $.
	\end{itemize}
	Furthermore, let $ \phi := \bigwedge_{v \in V} \always \sometime~v $. The construction is now complete, and is polynomial to the size of $ \tuple{G, c} $.

	Now, consider the smallest $ \cost(\kappa), \kappa \in \WI(\Game,\phi,\beta) $. We argue that $ \cost(\kappa) $ is indeed the lowest value such that a tour in $ G $ is attainable. Suppose for contradiction, that there exists $ \kappa' $ such that $ \cost(\kappa') < \cost(\kappa) $. Let $ \pi' $ be a path in $ (\Game,\kappa') $ and $ z_1 = \wFun_1(\tuple{\sink,\varepsilon}) $ the largest value player 1 can get by deviating from $ \pi' $. We have $ \pay_1(\pi') < z_1 $, and since for every $ \tuple{v,e} \in \St $ there exists an edge to $ \tuple{\sink,\varepsilon} $, thus player 1 would deviate to $ \tuple{\sink,\varepsilon} $ and stay there forever. This deviation means that $ \phi $ is not satisfied, which is a contradiction to $ \kappa' \in \WI(\Game,\phi,\beta) $. The construction of $ \phi $ also ensures that the path is a valid tour, i.e., the tour visits every city at least once. Notice that $ \phi $ does not guarantee a Hamiltonian cycle. However, removing the condition of visiting each city \textit{only once} does not remove the hardness, since \textit{Euclidean} \textsc{TSP} is \np-hard \cite{GGJ1976,PAPADIMITRIOU1977237}. Therefore, in the planar case there is an optimal tour that visits each city only once, or otherwise, by the triangle inequality, skipping a repeated visit would not increase the cost. Finally, since $ \WI(\Game,\phi,\beta) $ is efficient, we have $ \beta $ to be exactly the value of the optimum tour in the corresponding \textsc{TSP Cost} instance.
\end{proof}

\paragraph*{Proof of Theorem \ref{thm:optsi-ltl}}
\begin{proof}
	The proof is analogous to that of Theorem \ref{thm:optwi-ltl}.
\end{proof}

\paragraph*{Proof of Theorem \ref{thm:optsi-gr1}}
\begin{proof}
	Membership uses arguments analogous to those in Theorem \ref{thm:optwi-gr1}. For hardness, we reduce $ \MWQSAT_2 $ to \optsi using the same techniques used in Theorem \ref{thm:strong-gr1} with few modifications. Given a \\$ \MWQSAT_2 $ instance $ \tuple{\psi(\setx,\sety),\cFun} $, we construct a game $ \Game $ and \GRone formula $ \phi $, such that the optimum budget $ \beta $ corresponds to the value of optimal solution to $ \tuple{\psi(\setx,\sety),\cFun} $. To this end, we may assume that $ \SI(\Game,\phi,\beta) $ is efficient and construct $ \Game $ with exactly the same rules as in Theorem \ref{thm:strong-gr1} except for the following:
	\begin{itemize}
		\item clearly the value of $ \beta $ is unknown,
		\item the initial weight for each state $ s \in \St $
		\begin{itemize}
			\item $ \wFun_1(s) = \frac{2}{3} $, if $ s[0] = \sink $,
			\item
			\[
			\wFun_1(s) =
			\begin{cases}
			- \cFun(s[0]) + 1, & \text{if } s[0] \in \term_j \setminus \sety \wedge s[0] ~\text{is not negated in}~t_j \\
			1, & \text{if } s[0] \in \term_j \setminus \sety \wedge s[0]~\text{is negated in}~t_j; \\
			\end{cases}
			\]
			\item $ \wFun_1(s) = 0 $, otherwise;
		\end{itemize}
		\item given a subsidy scheme $ \kappa $, we update the weight for each $ s \in \St $ such that $ s[0] \in \term_j \setminus \sety $, that is, an $ x $-literal that appears in term $ t_j $
		\[
		\wFun_1(s) =
		\begin{cases}
		\wFun_1(s) + \kappa(s), & \text{if } s[0]~\text{is not negated in}~t_j \\
		\wFun_1(s), & \text{otherwise;}
		\end{cases}
		\]
	\end{itemize}
	the construction is complete and polynomial to the size of $ \tuple{\psi(\setx,\sety),\cFun} $.
	
	Let $ o $ be the optimal solution to $ \MWQSAT_2 $ given the input $ \tuple{\psi(\setx,\sety),\cFun} $. We claim that $ \beta $ is exactly $ o $. To see this, consider the smallest $ \cost(\kappa)$, $ \kappa \in \SI(\Game,\phi,\beta) $. We argue that this is indeed the least total weight of an assignment $ \vecx $ such that $ \psi(\setx,\sety) $ is true for every $ \vecy $. Assume towards a contradiction that $ \cost(\kappa) < o $. By the construction of $ \wFun_1(\cdot) $, there exists no $ \pi $ such that $ \pay_1(\pi) > \frac{2}{3} $. Therefore, any run $ \pi' $ that ends up in $ \sink $ is sustained by Nash equilibrium, which is a contradiction to $ \kappa \in \SI(\Game,\phi,\beta) $. Now, since $ \SI(\Game,\phi,\beta) $ is efficient, by definition, there exists no $ \kappa' \in \SI(\Game,\phi,\beta) $ such that $ \cost(\kappa') < \cost(\kappa) $. Thus we have $ \beta $ equals to $ o $ as required.
\end{proof}

\newpage 

\section{Algorithms}\label{apd:algos}


\begin{algorithm}[h]
	\caption{\label{alg:weak} \weak.}
	\textbf{Input}: A game $\Game$, a specification formula $\varphi$, and budget $ \beta $.
	
	\For{$ \kappa \in \K(\Game,\beta), s \in \St,$ {\normalfont and} $ (\strpElm_{-1},\dots,\strpElm_{-n}) \in \bigtimes_{j \in \Ag} (\times_{i \in \Ag \setminus {j}} \strElm_i) $}{
		Compute $ (\Game,\kappa)$\\
		\For{$ i \in \Ag $}{
		Compute $ z_i = \pun_i(s) $ using $ \strpElm_{-i} $
		}
		Compute $ (\Game,\kappa){[z]} $\\
			\If{{\normalfont there is} $ \strpElm \in \NE((\Game,\kappa){[z]}) $ {\normalfont such that} $ \pi(\strpElm) \models \phi $} {
				\Return $\mthfun{Accept}$
			}

	}	
	
	\Return $\mthfun{Reject}$
\end{algorithm}


\begin{algorithm}[h]
	\caption{\label{alg:strong} \strong.}
	
	\textbf{Input}: A game $\Game$, a specification formula $\varphi$, and budget $ \beta $.
	
	\For{$ \kappa \in \K(\Game,\beta)$}{
		Compute $ (\Game,\kappa)$\\
		\For{$ i \in \Ag $ {\normalfont and} $ s \in \St$}{
			Compute $ \pun_i((\Game,\kappa)) $
		}
		
		$f_{\exists} \gets \bot $;
		$f_{\forall} \gets \top $\\
		
		\For{${z} \in \{\pun_i(s): s \in \St\}^{\Ag}$}{
			Compute $(\Game,\kappa){[z]}$\\
			
			\If{{\normalfont there exists} $ \pi \in (\Game,\kappa){[z]}$ {\normalfont such that for each} $ i \in \Ag, \pay_i(\pi) \geq z_i $}{
				$f_{\exists} \gets \top $\\
			}
		}
		
		\For{${z} \in \{\pun_i(s): s \in \St\}^{\Ag}$}{	
			\If{{\normalfont there exists} $\pi \in (\Game,\kappa){[z]}$ {\normalfont such that for each} $ i \in \Ag, \pay_i(\pi) \geq z_i \wedge \pi \models \lnot \phi $}{
				$f_{\forall} \gets \bot $\\
			}
		}
		\If{$ (f_{\exists} \wedge f_{\forall}) $}{\Return $ \mthfun{Accept} $}
		
	}	
	
	\Return $\mthfun{Reject}$
\end{algorithm}

\end{document}